\documentclass[sigconf]{acmart}
\usepackage{algorithm}
\usepackage{algorithmic}
\usepackage{dsfont}
\usepackage{makecell}
\usepackage{enumitem}
\usepackage{multirow}
\usepackage{booktabs}
\usepackage{geometry}
\usepackage{subcaption}
\usepackage{amsmath}
\AtBeginDocument{%
  }

\begin{document}

\title{TFPS: A Temporal Filtration-enhanced Positive Sample Set Construction Method for Implicit Collaborative Filtering}

\renewcommand{\shortauthors}{Trovato et al.}

\begin{abstract}
  The negative sampling strategy can effectively train collaborative filtering (CF) recommendation models based on implicit feedback by constructing positive and negative samples. However, existing methods primarily optimize the negative sampling process while neglecting the exploration of positive samples. Some denoising recommendation methods can be applied to denoise positive samples within negative sampling strategies, but they ignore temporal information. Existing work integrates sequential information during model aggregation but neglects time interval information, hindering accurate capture of users' current preferences. To address this problem, from a data perspective, we propose a novel temporal filtration-enhanced approach to construct a high-quality positive sample set. First, we design a time decay model based on interaction time intervals, transforming the original graph into a weighted user-item bipartite graph. Then, based on predefined filtering operations, the weighted user-item bipartite graph is layered. Finally, we design a layer-enhancement strategy to construct a high-quality positive sample set for the layered subgraphs.
  We provide theoretical insights into why TFPS can improve Recall@k and NDCG@k, and extensive experiments on three real-world datasets demonstrate the effectiveness of the proposed method. Additionally, TFPS can be integrated with various implicit CF recommenders or negative sampling methods to enhance its performance.
\end{abstract}



\author{Jiayi Wu}
\affiliation{
	\institution{Beijing Institute of Technology}
	\city{Beijing}
	\country{China}
}

\author{Zhengyu Wu}
\affiliation{
	\institution{Beijing Institute of Technology}
	\city{Beijing}
	\country{China}
}

\author{Xunkai Li}
\affiliation{
	\institution{Beijing Institute of Technology}
	\city{Beijing}
	\country{China}
}

\author{Rong-Hua Li}
\email{lironghuabit@126.com}
\affiliation{
	\institution{Beijing Institute of Technology}
	\city{Beijing}
	\country{China}
}
\authornote{Corresponding author.}

\author{Guoren Wang}
\affiliation{
	\institution{Beijing Institute of Technology}
	\city{Beijing}
	\country{China}
}

\maketitle

\section{Introduction}
Implicit collaborative filtering (CF) \cite{1}, as an important technology in recommendation systems, is widely applied in data scenarios based on implicit feedback, such as users' listening \cite{2,3}, purchasing \cite{5}, and viewing \cite{6} behaviors. These implicit feedback differs from traditional explicit feedback \cite{9} (such as ratings) as they lack clear negative feedback signals, which poses challenges for model training. To overcome this issue, negative sampling strategies are widely used. They construct negative samples from uninteracted items to assist the model's learning, enabling it to better distinguish between items of interest and those not of interest to the user.

Existing negative sampling methods typically construct high-quality negative samples by designing complex negative sampling strategies. These methods often include techniques such as mixup \cite{10,11} and decoupling \cite{12, 13}, aiming to improve the model's discriminative ability with more challenging negative samples. However, these methods neglect the exploration of positive samples. In implicit CF, a user's positive actions (positive samples) are not only a direct reflection of their true interests but also the sole source of information for the model to learn user preferences, providing clear supervision signals and optimization directions during the model's training process \cite{14}. Additionally, positive samples are a crucial basis for constructing negative sample selection criteria. This is because negative samples are often obtained based on their similarity to positive samples, as seen in methods like DNS \cite{15}, MixGCF \cite{10}, and AHNS \cite{16}. Therefore, the higher the quality of the positive samples, the better the effectiveness of negative sampling, which in turn improves model performance. 

Since negative sampling methods do not explicitly optimize the positive set, we can leverage denoising recommendation techniques to denoise positive samples. They assign lower weights to noisy positive samples \cite{44,45,46,57} or drop these noisy positive samples \cite{44,47,48} based on statistical characteristics, such as the mean and variance of the loss. The above methods ignore temporal information, yet prior studies have demonstrated that user preferences evolve dynamically over time \cite{51}, such as short-term seasonal shifts (e.g., from summer t-shirts to winter coats) and long-term interest changes (e.g., from pop to classical music). Consequently, neglecting temporal information prevents model from effectively capturing current user preferences, thereby reducing its accuracy in predicting future interactions. As shown in Fig.~1, the performance of existing methods on timestamp-partitioned is significantly lower than on randomly partitioned datasets. This further demonstrates that these methods cannot effectively capture users' current preferences. STAM \cite{51} incorporates item sequence information based on interaction timestamps into the model aggregation process, enhancing model performance. However, it ignores time intervals, potentially misinterpreting distant interactions as recent preferences or dense short-term interactions as outdated preferences. By filtering out items with interaction intervals exceeding 30 days when constructing sequences in STAM, we observed that STAM's recall@20 improved by 4.04\% and 7.63\% on AmazonCDs and Lastfm datasets. This demonstrates that time interval information effectively enhances model’s ability to learn users’ current preferences. 

\begin{figure}[t]
	\centering
	\includegraphics[width=\linewidth]{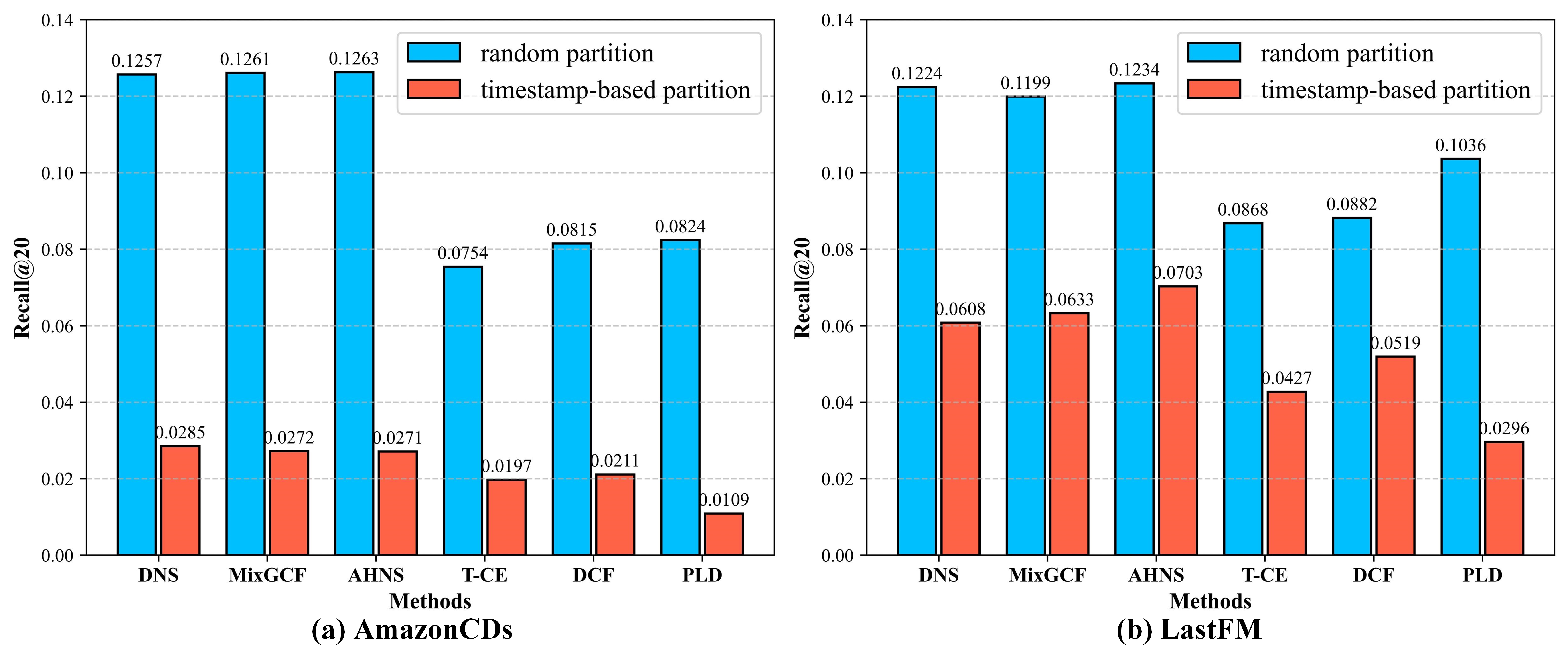}
	
	\vspace{-0.3cm}
	
	\caption{Impact of data partitioning on model performance.}
	
	\vspace{-0.5cm}
	
\end{figure}

Building on these observations, we propose a novel temporal filtration-enhanced positive sample set construction method (TFPS) aimed at building a high-quality positive sample set to improve the model's ability to learn users' current preferences. First, we apply a time-decay model to weight user interactions and transform the original bipartite graph into a weighted graph. This method assigns different weights based on the time intervals of interactions, thus better reflecting the user's current interests. Then, based on predefined graph filtration operations, we convert the weighted bipartite graph into multiple layered subgraphs, which reflect changes in user interests across different items over various time periods. Finally, we employ a layer-enhancement mechanism to construct a positive sample set for these layered subgraphs. The higher the weight of a subgraph, the larger the proportion of its edges in the positive sample set. This strategy prioritizes recent user interactions in the construction of the positive sample set, ensuring that model accurately captures user's current preferences, thereby enhancing the timeliness and accuracy of the recommendation system.

\textbf{Our contributions.}
(1) \textit{\underline{New Perspective}}. 
We revisit negative sampling from the positive-side perspective and, to our knowledge, are the first to incorporate temporal information into positive sample set construction for negative-sampling-based implicit CF, enabling the model to accurately capture users' current preferences.
(2) \textit{\underline{Simple yet Effective Approach}}. 
We propose a temporal filtration-enhanced positive sample set construction approach, TFPS. It employs graph filtration to extract subgraphs reflecting users' current preferences, using data-level reweighting (via duplicating recent positives) to guide the model in accurately learning these preferences. We theoretically analyze how training with this set can amplify expected margin gains, providing insight into improved Recall@k and NDCG@k.
(3) \textit{\underline{High Performance and Applicability}}. 
Extensive experiments demonstrate that TFPS not only outperforms state-of-the-art baselines, but can also be plugged into different implicit CF recommenders and used in conjunction with different negative sampling strategies to further improve model performance.

\section{Preliminaries}

\subsection{Backgrounds}

\noindent \textbf{Graph filtration}~\cite{17,18} is an important concept in graph analysis, which constructs a nested sequence of subgraphs by gradually including nodes/edges according to a predefined filtration criterion (e.g., edge weights or node degrees).
Given a weighted graph $G=(V,E,W)$, we illustrate edge-weight filtration by selecting $n$ increasing thresholds within $[W_{\min}, W_{\max}]$, i.e., $W_{\min}<w_1<\cdots<w_n < W_{\max}$.
For each threshold, we keep all edges whose weights are no larger than the threshold and form a filtered subgraph; as the threshold becomes less restrictive, more edges are included, yielding a nested sequence of subgraphs.
In practice, we further derive incremental layers from this nested filtration by taking the successive differences between adjacent thresholded subgraphs (i.e., newly added edges), thereby constructing a set of non-overlapping layered subgraphs that characterize the graph structure at different weight levels~\cite{19}.
In our work, these graph filtration-induced layers directly support the subsequent layer-enhancement strategy for constructing the positive sample set.

\textbf{Negative sampling} ~\cite{20} is a widely used technique in implicit CF recommendation, where items actually interacted with by the user are first treated as positive samples. Specific negative sampling strategies are then employed to obtain negative samples based on the items not interacted by users. Finally, the difference between the scores of positive and negative samples is maximized to improve the model's ability to identify the user's true interests, thereby improving the accuracy of the recommender.

\subsection{Existing Methods and Their Defects}

Existing negative sampling methods primarily focus on obtaining high-quality negative samples and are called negative sample-focused negative sampling methods. These methods are mainly divided into two categories: Select Negative Sample Methods and Construct Negative Sample Methods.

\textbf{Select Negative Sample Methods} primarily involve designing an appropriate negative sampling distribution to choose negative samples from uninteracted items for model training. For example, RNS \cite{22} designs the negative sampling distribution as a uniform distribution, randomly selecting negative samples from uninteracted items. PNS \cite{23} designs the negative sampling distribution based on item popularity, where more popular items have a higher probability of being selected as negative samples. DNS \cite{15} designs a negative sampling distribution based on the current model's predicted scores, where negative samples with higher scores have a greater chance of being selected. The negative sampling distribution in IRGAN \cite{24} is designed by the generator during the adversarial training process, where the generator selects negative samples that are more difficult for the discriminator to distinguish. SRNS \cite{25} combines a variance-based sampling strategy, selecting negative samples with large score fluctuations to avoid the risk of false negatives. The negative sampling distribution in AHNS \cite{16} is related to the positive sample score, the lower the positive sample score, the higher the probability of selecting more challenging negative samples. DNS(M,N) \cite{26} is an extended DNS, where two predefined hyper-parameters control the difficulty of selecting negative samples, resulting in a more flexible negative sampling distribution.

\textbf{Construct Negative Sample Methods} do not directly select negative samples from the set of uninteracted items, but instead create more challenging negative samples through operations in the vector space. For example, MixGCF \cite{10} first injects positive sample information into negative samples through positive mixing, then aggregates multi-hop neighbor information via hop mixing to generate negative samples. RecNS \cite{11} first uses Positive-Assisted Sampling to incorporate interaction information between users and positive samples to construct hard negative samples, then applies Exposure-Augmented Sampling to sample exposed but uninteracted data as negative samples, and finally synthesizes the negative samples constructed by both modules in the vector space. DENS \cite{13} first disentangles the relevant and irrelevant factors between positive samples and candidate negative samples in the vector space, then samples candidate negative samples that are similar to the positive samples in irrelevant factors but significantly different in relevant factors to serve as negative samples. ANS \cite{27} first disentangles the negative sample vectors into hard and easy factors, then augments the easy factor by controlling the direction and magnitude of noise to obtain negative samples. DBNS \cite{28} constructs negative samples using a diffusion model in the vector space, and controls the difficulty of negative samples based on time steps.

The above methods primarily focus on obtaining high-quality negative samples, while neglecting the exploration of positive samples and failing to consider the role of positive samples in improving model learning. Denoising recommendation methods can be applied to denoise positive samples within negative sampling strategies, called positive sample-focused negative sampling methods. These methods are mainly divided into two categories: Reweight Positive Sample Methods and Drop Positive Sample Methods.

\textbf{Reweight Positive Sample Methods} refer to assigning lower weights to noisy positive samples during the training process. For example, R-CE \cite{44} considers that positive samples with higher losses during training are more likely to be noisy samples, thus assigning lower weights to these samples. DeCA \cite{45} proposes a method for identifying noisy samples by leveraging the prediction consistency of different recommendation models. SGDL \cite{49} proposes a two-stage training approach. Memory data collected during the noise-resistant stage is used for self-guided training in the noise-sensitive stage, alternately optimizing the weight matrix parameters and the recommendation model parameters. Additionally, the authors adaptively determine the stage transition point by estimating the noise rate using a Gaussian Mixture Model \cite{50}. BOD \cite{46} introduces a bi-level optimization architecture to iteratively optimize the weight matrix and the model, ensuring that the weight matrix can store information from previous training. PLD \cite{57} designs a resampling strategy based on personal loss distribution to reduce the probability of noisy samples being optimized.

\begin{figure*}[t]
	\centering
	\includegraphics[width=0.93\linewidth]{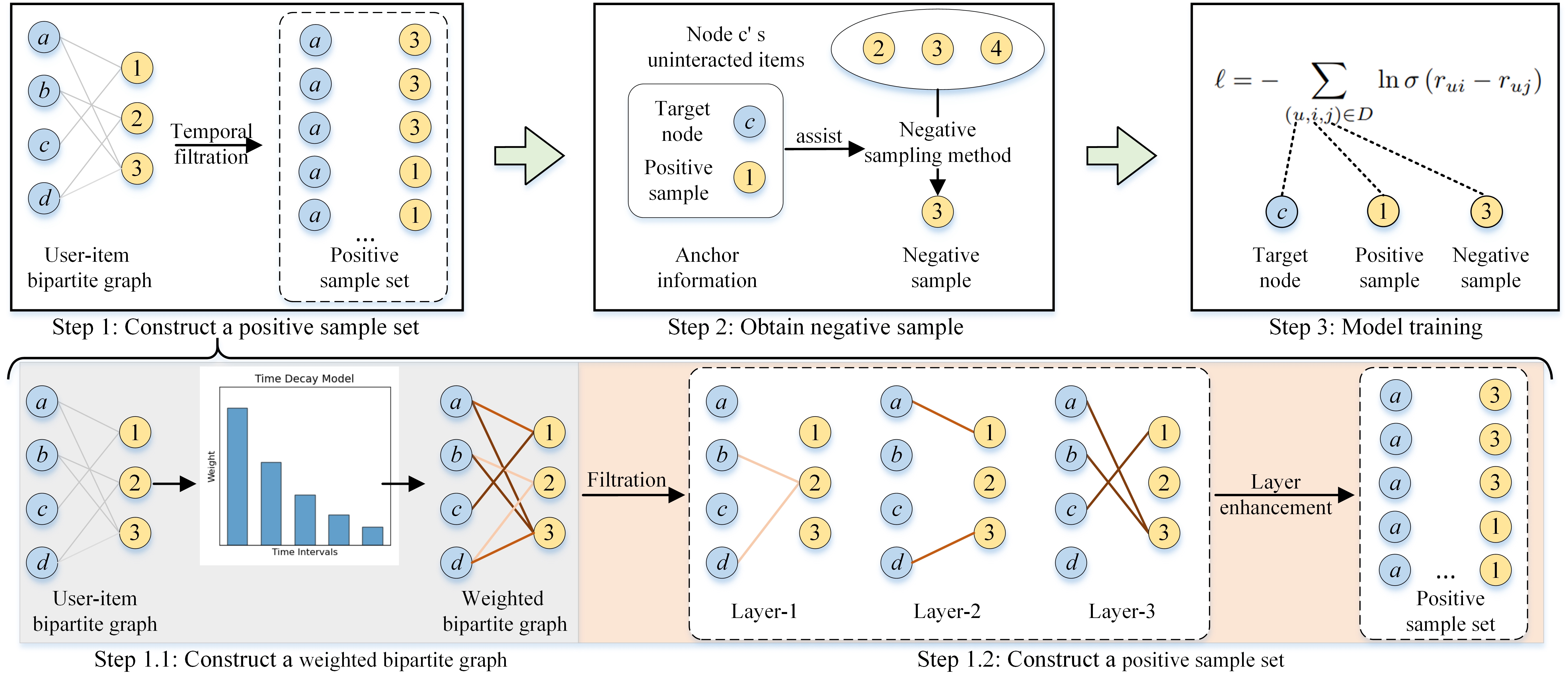}
	
	\vspace{-0.3cm}
	
	\caption{An overview of TFPS, where the darker the color of the edges, the higher the edge weight.}
	
	\vspace{-0.3cm}
	
\end{figure*}

\textbf{Drop Positive Sample Methods} involve dropping noisy positive samples during the training process. For example, T-CE \cite{44} drops positive samples with higher loss values during training. RGCF \cite{47} calculates the reliability of interactions based on the structural features of nodes and prunes interactions with lower reliability. DCF \cite{48} believes that the mean of losses cannot accurately distinguish between noisy samples and hard samples, so it combines the variance of losses to remove noisy samples; Then, a progressive label correction strategy is designed to correct some noisy samples, preventing the dataset from becoming sparse.

The aforementioned methods neglect temporal information, while user preferences continuously change over time, making it difficult for implicit CF recommendation models to capture current user interests, thus affecting their predictive performance for future interactions (See Fig.~1).  Sequential recommendation models, such as GRU4Rec \cite{60}, BERT4Rec \cite{61}, FEARec \cite{63}, and SIGMA \cite{65}, capture temporal dependencies by modeling user interaction sequences to predict the next interaction item, optimizing local behavior predictions. These sequential models differs fundamentally from the information modeling and optimization objectives of the models studied in this paper, and thus, we do not delve into sequential models. Instead, this paper focuses on implicit CF recommendation models, such as MF-BPR \cite{21}, NGCF \cite{22}, PinSage \cite{56}, and LightGCN \cite{32}, which model user-item interaction pairs and optimize global preference rankings through negative sampling strategies. Few studies explore integrating temporal information into this class of models. STAM \cite{51}, from a model perspective, constructs dynamic adjacency matrices based on the sequential order of user interactions, thereby incorporating temporal information during the aggregation phase. Nevertheless, this method neglects time interval information, potentially leading the model to erroneously treat sparse long-term interactions as recent interests or dense short-term interactions as outdated preferences, thus impacting the model’s accuracy in predicting future interactions. Driven by this issue, we innovatively propose a temporal filtration-enhanced negative sampling method from a data perspective, incorporating temporal information into negative sampling strategies. Theoretical analysis and extensive experiments demonstrate that TFPS can accurately capture users' current preferences and improve model performance.

\section{Methodology}

In this section, we introduce the proposed method TFPS. As shown in Fig.~2, TFPS focuses primarily on constructing a high-quality positive sample set. First, the original bipartite graph is transformed into a weighted bipartite graph using a time-decay model. Then, through graph filtration operations, the weighted bipartite graph is decomposed into multiple layers of bipartite subgraphs. Finally, we employ a layer-enhancement strategy based on the subgraph-level weight distribution to construct a high-quality positive sample set.

\subsection{Construct the Weighted Bipartite Graph}

Given an implicit feedback dataset $D = \big\{ {\big( {u,p,t_p^u} \big)|u \in U,p \in P} \big\}$, where $U$ and $P$ represent the set of users and the set of items, respectively. $\big( {u,p,t_p^u} \big)$ denotes that user $u$ interacted with item $p$ at time $t$.
We construct a user-item bipartite graph $G = \left\{ {V,E} \right\}$ based on dataset $D$, where $V$ represents the set of user and item nodes, and $E$ represents the set of user-item interaction behaviors.

Considering that users' interests typically evolve over time, recent behaviors more accurately reflect their current interest states. Moreover, these changes usually happen gradually rather than abruptly \cite{29}, so we adopt an simple yet effective exponential time-decay function to weight the interaction data and introduce one parameter to control the decay rate. 
For any interaction data $\big( {u,p,t_p^u} \big) \in D$, the formula for calculating the weight between $u$ and $p$ is as follows:
\begin{equation}
	W\left( {u,p} \right) = {e^{ - \lambda \left( {t_u^{max} - t_u^p} \right)}},
\end{equation}
where $t_u^{max}$ represents the most recent interaction time of user $u$, and $\lambda$ is the hyperparameter controlling the time decay rate.

\subsection{Construct the Positive Sample Set}

The time decay model assigns high weights to recent user interactions. To ensure the model focuses more on these high-weight interactions, we first introduce a graph filtration operation to construct filtration-enhanced layered subgraphs, extracting the high-weight edges most relevant to the user's current interests. We then use a layer-enhancement approach to construct the positive sample set, ensuring that higher-weighted layered subgraphs occupy a larger proportion in the positive sample set. This allows the model to focus more on the high-weight interactions that represent the user's current interests during the training process, thereby improving its ability to learn users' current preferences.

\subsubsection{Layered Subgraph Construction}

In the weighted user-item bipartite graph $\hat{G}$, edge weights lie in $[0,1]$.
Therefore, following prior work~\cite{19}, we set a threshold sequence from $0$ to $1$ by evenly partitioning $[0,1]$ into $n$ equal-length bins, yielding $0=w_0 < w_1 < \cdots < w_{n-1} < w_n=1$. 
This uniform discretization keeps the non-uniformity of the interaction-weight distribution solely governed by the decay coefficient $\lambda$, while $n$ controls the granularity of discretization.
With this threshold sequence, we induce an edge-weight filtration over $\hat{G}$ and obtain incremental layers between adjacent thresholds.
For each sub-interval $[w_i, w_{i+1})$ and any edge $e \in E$, we construct a corresponding layered subgraph $\hat{G}_i$, which retains only the edges in $\hat{G}$ whose weights fall within this interval (i.e., $w_i \le w_e < w_{i+1}$).
This layered processing approach enables the model to flexibly select and emphasize high-weight edges that are related to the users' current preferences when constructing the positive sample set.

\subsubsection{Layer Enhancement}

For high-layer subgraphs (i.e., those containing high-weight edges), we increase their proportion in the positive sample set $PSS$ so that the model focuses more on interactions that better reflect users' current preferences. However, overly emphasizing high-weight edges may over-concentrate the positive frequency distribution and harm generalization. To make this enhancement effective yet controllable, we adopt a linear layer-enhancement scheme: for any edge $e \in E$, if $e$ belongs to the $i$-th layered subgraph $\hat{G}_i$, we include it in $PSS$ with $i$ copies, thereby increasing the training exposure of high-weight interactions. Notably, $PSS$ is an augmented positive sample set that allows duplicated interaction pairs rather than a strictly de-duplicated set. The impact of enhancement intensity on performance is analyzed in Section~4.3.

Algorithm 1, which outlines the overall process of TFPS, consists of three parts. The first part (Line 1) constructs a weighted bipartite graph based on the time decay model. The second part (Lines 2-7) performs graph filtration on $\hat{G}$ to obtain layered subgraphs, while the third part (Lines 8-11) uses a layer enhancement mechanism to construct the positive sample set.

\subsection{Model Optimization}

We train the recommender with the Bayesian Personalized Ranking (BPR) loss~\cite{21}. 
For each interaction pair $(u,p)\in PSS$, we sample a negative item $p^- \sim f(\cdot \mid u,p)$ and define the loss as:
\begin{equation}
	\ell_{\mathrm{BPR}}(u,p,p^-)
	= - \ln \sigma \left( \mathbf{e}_u^\top \mathbf{e}_p - \mathbf{e}_u^\top \mathbf{e}_{p^-} \right),
\end{equation}
where $\sigma(\cdot)$ is the sigmoid function and $f(\cdot \mid u,p)$ denotes the negative sampling distribution.

\textbf{Remark.}
Training over $PSS$ induces an implicit discrete reweighting.
Let $m_{up}$ denote the number of occurrences of $(u,p)$ in $PSS$, then the positive distribution is
$\pi(u,p)=\frac{m_{up}}{\sum_{(u',p')\in PSS} m_{u'p'}}$,
and minimizing the BPR loss over $PSS$ is equivalent to minimizing
$\mathbb{E}_{(u,p)\sim \pi}\mathbb{E}_{p^-\sim f(\cdot\mid u,p)}[\ell_{\mathrm{BPR}}(u,p,p^-)]$.
This implicit reweighting is intentionally realized at the data level via frequency redistribution: by repeating high-weight (recent) interactions, TFPS allocates more optimization steps to signals reflecting users' current preferences, while still retaining historical interactions to preserve informative long-term preferences.
Moreover, graph filtration discretizes continuous time-decay weights into a hierarchy of subgraphs, whose granularity is controlled by the number of layers $n$, enabling fine-grained control over which interactions are emphasized in $PSS$.
In contrast, Weighted-BPR (continuous loss reweighting) keeps the original positive sample set unchanged and only rescales per-instance gradients, which may still provide limited training exposure to recent positives; when weights are highly skewed, gradients of older interactions can become near-zero, diminishing long-term preference signals and potentially hurting model performance, as evidenced by the results in Section~4.3.

\begin{algorithm}[t]
	\renewcommand{\algorithmicrequire}{\textbf{Input:}}
	\renewcommand{\algorithmicensure}{\textbf{Output:}}
	\caption{Positive Sample Set Construction}
	\label{alg}
	\begin{algorithmic}[1]
		
		\REQUIRE Training set ${D_{train}} = \big\{ {\big( {u,p,t_p^u} \big)} \big\}$, Time decay rate $\lambda$, Number of subgraph layers $n$ \\ 
		\ENSURE Positive Sample Set ($PSS$) \\
		\STATE $\hat{G} = \left\langle V, E, W \right\rangle \leftarrow$ Construct the weighted bipartite graph from $D_{\text{train}}$ using (1).

		\STATE $W_{\text{max}} = \max(W)$, $W_{\text{min}} = \min(W)$

		\STATE $\text{Interval}_n \leftarrow$ Divide $[W_{\text{min}}, W_{\text{max}}]$ into $n$ equal intervals

		\STATE $G_1, G_2, \dots, G_n \leftarrow$ Initialize $n$ layered subgraphs
		
		\FOR {each $(u, p) \in E$}
		\STATE Add $(u, p)$ to the i-th subgraph $G_i$ if $W(u, p) \in [W_{i-1}, W_i)$
		\ENDFOR
		
		\STATE $PSS \leftarrow \emptyset$
		\FOR{$i = 1$ to $n$}
		\STATE $PSS \leftarrow PSS \cup (E_i \times i)$ \COMMENT{Add edges from $E_i$ repeated $i$ times}
		\ENDFOR
		\RETURN $PSS$
		
	\end{algorithmic}
\end{algorithm}

\subsection{Theoretical Analysis}

In this section, we provide theoretical insights into TFPS by analyzing how its data-level reweighting (via duplicating recent positives) amplifies pairwise margin improvements.
These results help explain the empirical gains of TFPS on top-$k$ ranking metrics such as Recall@k and NDCG@k under timestamp-partitioned evaluation.
\begin{theorem}[Local one-step margin improvement and TFPS-induced update amplification]
	\label{thm:margin_gain}
	Consider an implicit CF model trained with the BPR loss and a first-order optimizer (e.g., Adam).
	At each step, a negative item $p^- \sim f(\cdot\mid u,p)$ is sampled using the current parameters before the subsequent update.
	Let $s_{up}(\theta)$ denote the model score of user $u$ on item $p$ under parameters $\theta$.
	Define the pairwise margin as $\Delta_{u,p,p^-}(\theta)\triangleq s_{up}(\theta)-s_{u p^-}(\theta)$.
	We analyze a one-step update under standard $L$-smoothness conditions in a stable-training regime, so that the Taylor remainder satisfies
	$|r_t|\le \tfrac{L}{2}\|\theta^{+}-\theta\|^2=O(\eta^2)$.
	Moreover, the update can be written as
	$
	\theta^{+}=\theta-\eta\, D_t\, \nabla_\theta \ell_{\mathrm{BPR}}(u,p,p^-),
	$
	where $\eta$ is the learning rate and $D_t$ is a positive diagonal matrix that element-wise rescales the gradient.
	Then, conditioning on a positive pair $(u,p)$, the expected one-step margin improvement satisfies
	\begin{equation}
		\begin{aligned}
			\mathbb{E}\!\left[\Delta_{u,p,p^-}(\theta^{+})-\Delta_{u,p,p^-}(\theta)\mid (u,p)\right]
			\ge
			\eta\,\mathbb{E}_{p^-\sim f(\cdot\mid u,p)}\!  \\ \big[
			\sigma\!\left(-\Delta_{u,p,p^-}(\theta)\right)
			\hspace*{0pt}\cdot\,
			\nabla_\theta \Delta_{u,p,p^-}(\theta)^\top
			D_t\,
			\nabla_\theta \Delta_{u,p,p^-}(\theta)
			\big]
			-\; O(\eta^2).
		\end{aligned}
	\end{equation}
	Moreover, TFPS performs data-level reweighting by duplicating high-weight (recent) interactions, thereby increasing their proportion in $PSS$.
	Let $m_{up}$ denote the number of occurrences of $(u,p)$ in $PSS$.
	When training iterates over $PSS$, the pair $(u,p)$ is used for parameter updates exactly $m_{up}$ times per epoch.
	Thus, high-weight interactions with larger $m_{up}$ obtain more margin-increasing updates in total (in expectation).
\end{theorem}
\begin{proof}[Proof sketch]
	Fix a positive pair $(u,p)$ and a sampled negative item $p^- \sim f(\cdot\mid u,p)$ in the current step.
	For brevity, denote $\Delta(\theta)\equiv \Delta_{u,p,p^-}(\theta)$.
	The one-step update is
	\begin{equation}
		\theta^{+}=\theta-\eta\,D_t\,\nabla_\theta \ell_{\mathrm{BPR}}(u,p,p^-).
	\end{equation}
	Since $\ell_{\mathrm{BPR}}(u,p,p^-)=-\log\sigma(\Delta(\theta))$, by the chain rule,
	\begin{equation}
		\nabla_\theta \ell_{\mathrm{BPR}}(u,p,p^-)
		= -\sigma\!\left(-\Delta(\theta)\right)\,\nabla_\theta \Delta(\theta).
		\label{eq:bpr_grad}
	\end{equation}
	By a first-order expansion with remainder under the $L$-smoothness condition, there exists a scalar $r_t$ such that
	\begin{equation}
		\Delta(\theta^{+})-\Delta(\theta)
		=
		\nabla_\theta \Delta(\theta)^\top(\theta^{+}-\theta) + r_t,
	\end{equation}
	where $|r_t|\le \tfrac{L}{2}\|\theta^{+}-\theta\|^2$.
	Substituting $\theta^{+}-\theta=-\eta D_t\nabla_\theta \ell_{\mathrm{BPR}}$ and~\eqref{eq:bpr_grad} yields
	\begin{equation}
		\Delta(\theta^{+})-\Delta(\theta)
		=
		\eta\,\sigma(-\Delta(\theta))\,
		\nabla_\theta \Delta(\theta)^\top D_t \nabla_\theta \Delta(\theta) + r_t.
	\end{equation}
	In the stable-training regime, $\|D_t\|$ and $\|\nabla_\theta \ell_{\mathrm{BPR}}\|$ are bounded, hence
	$\|\theta^{+}-\theta\|=\eta\|D_t\nabla_\theta \ell_{\mathrm{BPR}}\|=O(\eta)$ and thus $r_t\ge -O(\eta^2)$.
	Therefore,
	\begin{equation}
		\Delta(\theta^{+})-\Delta(\theta)
		\ge
		\eta\,\sigma(-\Delta(\theta))\,
		\nabla_\theta \Delta(\theta)^\top D_t \nabla_\theta \Delta(\theta)
		-\; O(\eta^2).
	\end{equation}
	Taking expectation over $p^- \sim f(\cdot\mid u,p)$ gives the stated inequality.
	The TFPS claim follows since iterating over $PSS$ uses $(u,p)$ for exactly $m_{up}$ updates per epoch,
	thereby accumulating $m_{up}$ such local margin-improving updates.
\end{proof}

\begin{corollary}[TFPS insight for test Recall@k via margin amplification]
	\label{cor:recall_insight}
	Given a positive interaction $(u,p)\in PSS$, with a negative item $p^- \sim f(\cdot\mid u,p)$ sampled for $(u,p)$.
	By Theorem~\ref{thm:margin_gain}, each update on $(u,p)$ increases the pairwise margin
	$\Delta_{u,p,p^-}=s_{up}-s_{u p^-}$ in expectation, and TFPS assigns $m_{up}$ update opportunities per epoch
	to high-weight interactions via duplication in $PSS$, thereby amplifying their expected cumulative separation from sampled negatives.
	Larger margins make sampled competitors less likely to outrank $p$, which tends to move $p$ upward in the ranked list produced by the learned parameters $\theta$.
	
	In timestamp-partitioned evaluation, the test positive set $P(u)$ is biased toward recent interests and thus is more aligned with the high-weight interactions emphasized by TFPS.
	Therefore, improving the ranks of such recent positives during training tends to translate into better ranks of items in $P(u)$ at test time.
	The test Recall@k for $u$ is
	\begin{equation}
		\mathrm{Recall@k}(u)
		=
		\frac{|P(u)\cap R_k(u)|}{|P(u)|}
		=
		\frac{1}{|P(u)|}\sum_{p\in P(u)} \mathbb{I}(r_p\le k),
	\end{equation}
	where $R_k(u)$ denotes the top-$k$ recommendation list, $r_p$ denotes the rank of $p$,
	and $\mathbb{I}(\cdot)$ is the indicator function.
	Hence, making test positives rank earlier (smaller $r_p$) increases the chance that $\mathbb{I}(r_p\le k)=1$ and thus improves test Recall@k, providing theoretical insight into the observed gains.
\end{corollary}

\begin{corollary}[TFPS insight for test NDCG@k via a margin-monotone lower bound]
	\label{cor:ndcg_insight}
		Given a user $u$ (using the notation in Corollary~\ref{cor:recall_insight}),
		we introduce a margin-based surrogate that measures how well a test positive $p$ is separated from competing items:
		\begin{equation}
			\phi_u(p;\theta)
			\triangleq
			\frac{1}{1+\sum_{q\neq p}\exp\!\left(s_{uq}(\theta)-s_{up}(\theta)\right)}.
		\end{equation}
		Then, the test NDCG@k admits the following lower bound:
		\begin{equation}
			\mathrm{NDCG@k}(u) = \frac{\text{DCG@k}(u)}{\text{IDCG@k}(u)}
			\ge
			\frac{1}{Z_k(u)}
			\sum_{p\in P(u)}
			\mathbb{I}(r_p\le k)\,\phi_u(p;\theta),
			\label{eq:ndcg_lb}
		\end{equation}
	where
	$\mathrm{DCG@k}(u)=\sum_{p\in P(u)\cap R_k(u)}\frac{1}{\log_2(r_p+1)}$ and
	$\mathrm{IDCG@k}(u)=Z_k(u)\triangleq \sum_{i=1}^{\min(k,|P(u)|)}\frac{1}{\log_2(i+1)}$.
	The inequality follows since for any $p\in P(u)\cap R_k(u)$,
	$\frac{1}{\log_2(r_p+1)}\ge \frac{1}{r_p}\ge \phi_u(p;\theta)$,
	using $\log_2(r_p+1)\le r_p$ and
	$r_p=1+\sum_{q\neq p}\mathbb{I}(s_{uq}>s_{up})\le 1+\sum_{q\neq p}\exp(s_{uq}-s_{up})$.

	By Theorem~\ref{thm:margin_gain}, TFPS allocates more margin-improving updates to high-weight (recent) positives during training,
	which tends to increase their pairwise margins $s_{up}(\theta)-s_{uq}(\theta)$ and hence increases $\phi_u(p;\theta)$.
	Under timestamp-partitioned evaluation, test positives are biased toward recent interests and thus are more aligned with these high-weight positives;
	therefore the right-hand side term $\mathbb{I}(r_p\le k)\,\phi_u(p;\theta)$ in~\eqref{eq:ndcg_lb} tends to increase,
	providing theoretical insight into the observed test NDCG@k improvements.
\end{corollary}

\subsection{Time Complexity}

The time complexity of TFPS for constructing the positive sample set comprises two modules. For the weighted bipartite graph construction module (as shown in Algorithm 1), it involves three steps: bipartite graph initialization, computation of users’ most recent interaction times, and edge weight calculation. Each step requires traversing all interaction edges, resulting in a complexity of $O(|E|)$, where $|E|$ denotes the number of interaction edges. For the positive sample set construction module (as shown in Algorithm 2), it includes two steps: constructing layered subgraphs and constructing the positive sample set. Constructing layered subgraphs via graph filtration incurs a complexity of $O(|E|)$, as the number of layers is a small constant. Constructing the positive sample set via layer enhancement requires traversing all edges in the subgraphs (equivalent to interaction edges), also with a complexity of $O(|E|)$. Thus, this module’s complexity is $O(|E|)$. The time complexity of TFPS is $O(|E|)$, linear in the number of interaction edges. Moreover, TFPS only needs to run once before model training. We show TFPS’s efficiency in constructing positive sample sets in Section~IV.B.

\section{Experiments}

To demonstrate the effectiveness of the TFPS method, we conducted extensive experiments and answered the following questions: \textbf{(RQ1)} How does TFPS perform in comparison to baselines focusing on negative samples and on positive samples in implicit CF models? \textbf{(RQ2)} How do the parameters of TFPS and the distribution of positive samples affect the model performance, and why is the layer enhancement mechanism preferred over recency-aware strategies? \textbf{(RQ3)} How does TFPS perform when integrated with other negative sampling methods and implicit CF models (e.g., matrix factorization \cite{30, 31})? \textbf{(RQ4)} How does TFPS perform in a reference comparison with sequential recommendation models?

%\item {(RQ1)} How does TFPS perform in comparison to other baselines in implicit CF recommendation models? 
%
%
%\item {(RQ2)} How does parameters of TFPS affect the performance of recommendation tasks? 
%
%
%\item {(RQ3)} How does TFPS perform when integrated into other implicit CF recommendation models (e.g., matrix factorization \cite{30, 31})? 
%
%\item {(RQ4)} How does TFPS perform when integrated with other negative sampling methods?

\subsection{Experimental Setup}

\subsubsection{Datasets}

We select three widely used real-world datasets in implicit CF research \cite{58, 59, 13}, Ta-Feng, LastFM, and AmazonCDs, to evaluate TFPS. Ta-Feng\footnote{https://www.kaggle.com/datasets/chiranjivdas09/ta-feng-grocery-dataset} contains users' purchase data from a supermarket in Taiwan. LastFM\footnote{https://grouplens.org/datasets/hetrec-2011/} contains users' interaction data with music from LastFM music platform, such as listening records and tags. AmazonCDs\footnote{https://cseweb.ucsd.edu//~jmcauley/datasets/amazon\_v2/index.html} is based on user reviews of CD products from Amazon e-commerce platform. 
TFPS is designed for negative sampling-based CF models. It injects temporal signals into the positive sample set so that these models (without modifying the model architecture) can learn from past interactions to predict future behaviors. Accordingly, we divide the dataset based on a timestamp: we sort all interactions chronologically and choose the cutting timestamp such that approximately 80\% of interactions occur before it as the training set, while the remaining interactions constitute the test and validation sets. We remove users who only appear in the test and validation sets. The dataset statistics are shown in Table I. We use the widely adopted metrics, Recall@k and NDCG@k, to evaluate the performance of the recommender, with k values set to \{20, 30\}.

\begin{table}[t]
	\centering
	
	\vspace{0.1cm}
	
	\caption{Dataset Statistics}
	
	\small
	
	\setlength{\tabcolsep}{3.5pt} % 数值越大，列间距越宽
	
	\begin{tabular}{cccccc}
		\toprule
		Dataset & Users & Items & Interactions & Density & \makecell{Cutting\\Timestamp} \\
		\midrule
		Ta-Feng   & 29,339  & 23,685  & 793,602  & 0.114\% & 2001-02-04 \\
		LastFM    & 1,665   & 9,121   & 176,103  & 1.160\% & 2010-07-31 \\
		AmazonCDs & 91,190  & 73,692  & 1,242,564 & 0.018\% & 2014-12-01 \\
		\bottomrule
	\end{tabular}
	
	\vspace{-0.1cm}
	
\end{table}

\begin{table*}[t]
	\centering
	\caption{Overall Performance Comparison.}
	
	\vspace{-0.1cm}
	
	\resizebox{\linewidth}{36mm} {
		\setlength{\tabcolsep}{2mm}{
			\begin{tabular}{c|cccc|cccc|cccc}
				\toprule
				\multirow{1}{*}{Dataset} & \multicolumn{4}{c|}{AmazonCDs} & \multicolumn{4}{c|}{LastFM} & \multicolumn{4}{c}{Ta-Feng} \\ 
				\cmidrule(r){1-1} \cmidrule(r){2-5} \cmidrule(r){6-9} \cmidrule(r){10-13}
				\multirow{1}{*}{Method} & R@20 & N@20 & R@30 & N@30 & R@20 & N@20 & R@30 & N@30 & R@20 & N@20 & R@30 & N@30 \\
				\midrule
				RNS       & 0.0270 & 0.0152 & 0.0354 & 0.0175 & 0.0649 & 0.0723 & 0.0777 & 0.0738 & 0.0516 & 0.0424 & 0.0636 & 0.0464 \\
				DNS       & \underline{0.0285} & \underline{0.0159} & \underline{0.0357} & \underline{0.0180} & 0.0608 & 0.0724 & 0.0792 & 0.0754 & 0.0440 & 0.0319 & 0.0552 & 0.0356 \\
				DNS (M,N) & 0.0264 & 0.0150 & 0.0344 & 0.0171 & 0.0610 & 0.0736 & 0.0825 & 0.0778 & 0.0430 & 0.0295 & 0.0537 & 0.0331 \\
				MixGCF    & 0.0272 & 0.0153 & 0.0348 & 0.0174 & 0.0633 & 0.0744 & 0.0786 & 0.0775 & 0.0465 & 0.0334 & 0.0582 & 0.0373 \\
				DENS      & 0.0270 & 0.0149 & 0.0347 & 0.0170 & 0.0737 & 0.0756 & 0.0917 & 0.0790 & 0.0475 & 0.0401 & 0.0558 & 0.0429 \\
				AHNS      & 0.0271 & 0.0152 & 0.0342 & 0.0172 & 0.0703 & 0.0780 & 0.0811 & 0.0793 & 0.0528 & 0.0439 & 0.0692 & 0.0493 \\
				\midrule
				R-CE      & 0.0174 & 0.0096 & 0.0233 & 0.0112 & 0.0504 & 0.0614 & 0.0617 & 0.0631 & 0.0558 & 0.0450 & 0.0662 & 0.0484 \\
				T-CE      & 0.0197 & 0.0109 & 0.0263 & 0.0127 & 0.0427 & 0.0548 & 0.0584 & 0.0568 & 0.0550 & 0.0452 & 0.0669 & 0.0491 \\
				DeCA      & 0.0150 & 0.0082 & 0.0210 & 0.0099 & 0.0532 & 0.0605 & 0.0644 & 0.0624 & 0.0551 & 0.0450 & 0.0675 & 0.0493 \\
				DCF       & 0.0211 & 0.0116 & 0.0278 & 0.0135 & 0.0519 & 0.0603 & 0.0618 & 0.0613 & 0.0555 & 0.0454 & 0.0666 & 0.0492 \\
				PLD       & 0.0109 & 0.0061 & 0.0140 & 0.0070 &  0.0296& 0.0251& 0.0356&  0.0259           &  0.0207& 0.0170 &  0.0254& 0.0186 \\
				\midrule
				STAM   & 0.0259 & 0.0144 & 0.0334 & 0.0165 & 0.0673 & 0.0732 & 0.0794 & 0.0759 & 0.0547 & 0.0418 & 0.0660 & 0.0456 \\
				TFPS-STAM   & 0.0281 & 0.0156 & 0.0351 & 0.0175 & \underline{0.1098} & \underline{0.0882} & \underline{0.1382} & \underline{0.0973} & \underline{0.0707} & \underline{0.0581} & \underline{0.0863} & \underline{0.0634} \\
				\midrule
				TFPS (Ours) & \textbf{0.0308} & \textbf{0.0173} & \textbf{0.0388} & \textbf{0.0195} & \textbf{0.2153} & \textbf{0.2300} & \textbf{0.2485} & \textbf{0.2395} & \textbf{0.0771} & \textbf{0.0590} & \textbf{0.0915} & \textbf{0.0638} \\
				\bottomrule
			\end{tabular}
		}
	}
	
	\vspace{-0.1cm}

\end{table*}

\subsubsection{Baselines}

Existing negative sampling baselines, whether focusing on positive or negative samples, overlook temporal information, limiting their ability to capture current user preferences. To validate the effectiveness of TFPS in incorporating temporal information into negative sampling, we compare it against six negative sample-focused negative sampling methods RNS~\cite{21}, DNS~\cite{15}, DNS(M,N)~\cite{26}, MixGCF~\cite{10}, DENS~\cite{13}, AHNS~\cite{16} and five positive sample-focused denoising recommendation methods T-CE~\cite{44}, R-CE~\cite{44}, DeCA~\cite{45}, DCF~\cite{48}, PLD~\cite{57}. Additionally, we also compare TFPS with STAM~\cite{51} (a spatiotemporal aggregation method for GNN-based CF), to demonstrate the differential impact of data-level versus model-level temporal information integration on recommendation model performance. We further evaluate \textit{TFPS-STAM}, which combines data-level temporal enhancement (TFPS) with model-level temporal aggregation (STAM). Unless otherwise specified, all baselines use the same hyperparameter settings as their original code implementations.

\textit{Remark.} To ensure a fair comparison, the pool size for all methods that use candidate negative sample pool techniques (DNS, MixGCF, DENS, AHNS) is set to 10. All experiments take data leakage into account, and data contained in the validation and test set are removed from the constructed positive sample set. Our source code is available at https://anonymous.4open.science/r/TFPS-03D2.

\begin{figure}[t]
	\centering
	\includegraphics[width=\linewidth]{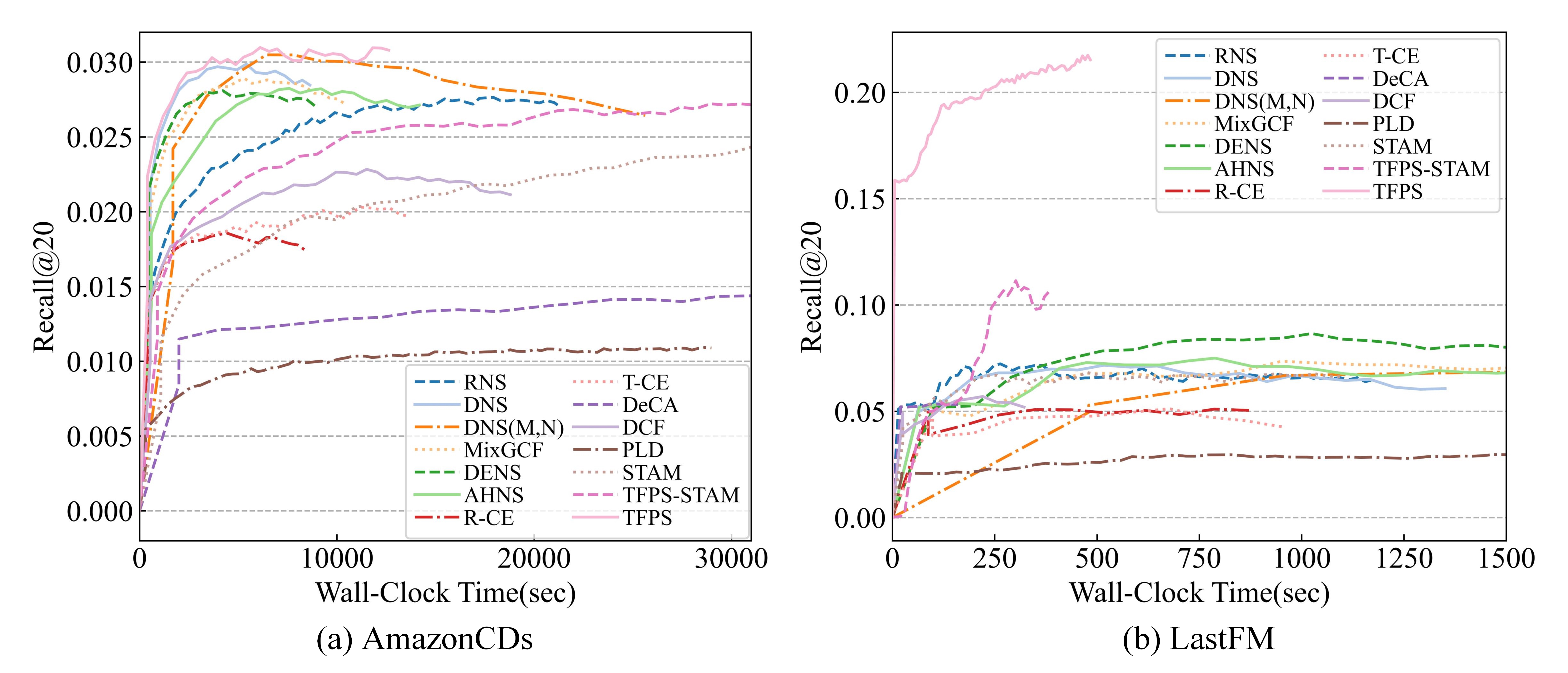}
	
	\vspace{-0.1cm}
	
	\caption{Recall@20 vs. wall-clock time (in seconds).}
	
	\vspace{-0.3cm}
	
\end{figure}

\subsubsection{Parameter settings}

All experiments are conducted on LightGCN \cite{32}. The initial vectors are initialized using the Xavier method \cite{41} with a dimension of 64, Adam optimizer \cite{42} is used with a learning rate of 0.001, mini-batch size is set to 2048, L2 regularization parameter to 0.0001, and the model layers is set to 3. We run TFPS with five different random seeds and report the averaged results.

\subsection{Performance Comparison (RQ1)}

We present the results of the comparison between TFPS and the other twelve methods across three datasets in Table II. The best results are highlighted in bold and the second-best results are underlined. We observe the following experimental phenomena:

[1] TFPS achieves the best results across all metrics on the different datasets. We attribute this to the following reasons: (1) The time decay model in TFPS assigns higher weights to users' recent interactions, and graph filtration can flexibly extract high-weighted subgraphs, effectively capturing users' recent preferences. (2) By focusing on constructing the positive sample set in the negative sampling strategy, TFPS allows the model to learn more accurately from the positive sample supervision signals, improving the timeliness and accuracy of recommendation model.

[2] STAM underperforms TFPS because it relies solely on positional embeddings to model the temporal information of user interactions, neglecting temporal intervals, which makes it challenging for model to distinguish between current and outdated preferences. In contrast, TFPS captures temporal interval information using a time-decay model, and constructs a positive sample set through graph filtration and layer enhancement, explicitly guiding model to learn current preferences. TFPS-STAM outperforms STAM, which can also be attributed to the explicit guidance provided by the positive sample set constructed by TFPS. 
Notably, TFPS-STAM performs worse than TFPS alone. TFPS-enhanced positives are more temporally concentrated and redundant, which weakens STAM’s modeling of informative temporal transitions and may partially undermine its position-based encoding. Consequently, STAM constructs a noisier adjacency matrix and yields lower accuracy.
Based on this finding, we conducted experiments to further explore the applicability of TFPS, revealing that it underperforms only when integrated with sequence-based temporal modeling methods. Integrating TFPS with other negative sampling strategies (see Fig.~6) and backbone (see Table III) significantly enhances its performance.

[3] These baselines exhibit significant performance differences across various datasets. For example, most negative-sample-focused baselines perform well on the AmazonCDs dataset but poorly on the Ta-Feng dataset; whereas most positive-sample-focused baselines shows the opposite result. We attribute this to differences in dataset characteristics (e.g., sparsity, noise levels). Notably, RNS outperforms other baselines in many cases. We speculate that this stems from challenges posed by user preference drift in timestamp-split datasets. Specifically, most negative-sample-focused baselines obtain negative samples based on local training set patterns (e.g., historical user behaviors), risking overfitting to early interactions. As user preferences evolve over time, this impairs the model’s accuracy in predicting future interactions. In contrast, RNS obtains negative samples without being confined to specific local patterns, better covering the dataset’s dynamic characteristics and enhancing model generalization. Most positive-sample-focused baselines filter noise based on statistical characteristics observed in specific training datasets, resulting in limited applicability. This leads to low denoising accuracy and consequently degrades model performance. Notably, PLD, the current state-of-the-art denoising method, performs particularly poorly on timestamp-split datasets. We attribute this to limitations in its positive sample resampling strategy. For instance, on the Lastfm dataset, resampling covers only 64.7\% of positive samples, classifying the remainder as noise, which significantly exacerbates the sparsity of positive supervision signals. Furthermore, we find that the average time interval between resampled positive samples and users’ most recent interactions reaches 140 days, severely impairing the model’s ability to capture users’ current preferences.

Besides effectiveness, we also compared the performance in terms of efficiency. We measured the time required for TFPS to construct positive sample sets on the Ta-Feng, LastFM, and AmazonCDs datasets, recording 22.10s, 3.03s, and 40.24s, respectively. These results align with the analysis in Section~III.E, which indicates that the time complexity of TFPS is linear with the number of interactions. Crucially, TFPS only needs to construct the positive sample set once before model training begins. In contrast, denoising recommendation methods T-CE, R-CE, DeCA, and DCF require frequent reconstruction of positive sample sets during the training process. Consequently, TFPS significantly outperform these denoising recommendation methods in terms of construction efficiency. Then, we compared the training efficiency of TFPS with that of the baselines, and the results are shown in Fig.~3. It can be observed that TFPS converges faster and achieves better performance.

\begin{figure}[t]
	\centering
	\includegraphics[width=\linewidth]{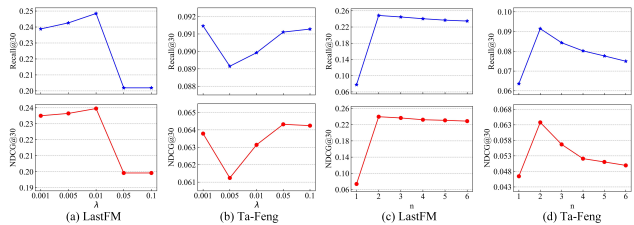}
	
	\caption{The impact of $\lambda$ and $n$ on Recall@30 and NDCG@30.}
	
	\vspace{-0.6cm}
	
\end{figure}

\subsection{Parameter Analysis and Ablation Study (RQ2)}

In this section, we analyze the impact of the parameters $\lambda$ and $n$ involved in TFPS on model performance. We search parameter $\lambda$ in range of \{0.001, 0.005, 0.01, 0.05, 0.1\}. Given that an excessively large $n$ may result in a more skewed positive sample frequency distribution and raise the risk of overfitting, we tune $n$ over $\{1,2,3,4,5,6\}$. When analyzing one parameter, the other parameters remain fixed.

Fig.~4(a),(b) show the effect of $\lambda$. As $\lambda$ increases, Recall@30 and NDCG@30 generally rise and then decline on the LastFM datasets. 
This is because a small $\lambda$ decays too slowly and allows high-layer subgraphs to mix in more historical-preference interactions, weakening the focus on current preferences. 
Conversely, when $\lambda$ is large, the time weights decay too quickly, which may cause the high-layer subgraph to miss informative recent interactions, reduce effective supervision, and degrade performance.
On the Ta-Feng dataset, a reverse trend is observed. We speculate this is because this dataset has a smaller time span and more focused user preferences, allowing a small $\lambda$ to effectively capture users' current preferences. As the $\lambda$ increases, graph filtration gradually shifts the segmentation of these interactions from a coarse-grained level to a fine-grained level. Coarse-grained segmentation may prevent the model from accurately focusing on the user's current interests, thereby affecting the model performance. Overall, TFPS remains superior to the strongest baseline for a broad range of $\lambda$ choices on both Recall@30 and NDCG@30.

Fig.~4(c),(d) show the effect of $n$. As $n$ increases, Recall@30 and NDCG@30 generally follow a trend of rising initially and then declining. This is because when $n$ is small, the weight intervals in graph filtration are larger, leading to larger subgraph sizes. High-layer subgraphs may include interactions reflecting the user's historical interests, thus affecting the model's ability to capture the user's current preferences. When $n$ is large, the weight intervals become smaller, resulting in smaller subgraph sizes, which may cause high-layer subgraphs to miss some interactions reflecting the user's current interests. Moreover, a larger $n$ also lead to model overfitting, reducing the model performance. Empirically, $n=2$ offers a good trade-off between capturing recent temporal preferences and avoiding overly concentrated positives that cause overfitting.

\begin{figure}[t]
	\centering
	\includegraphics[width=\linewidth]{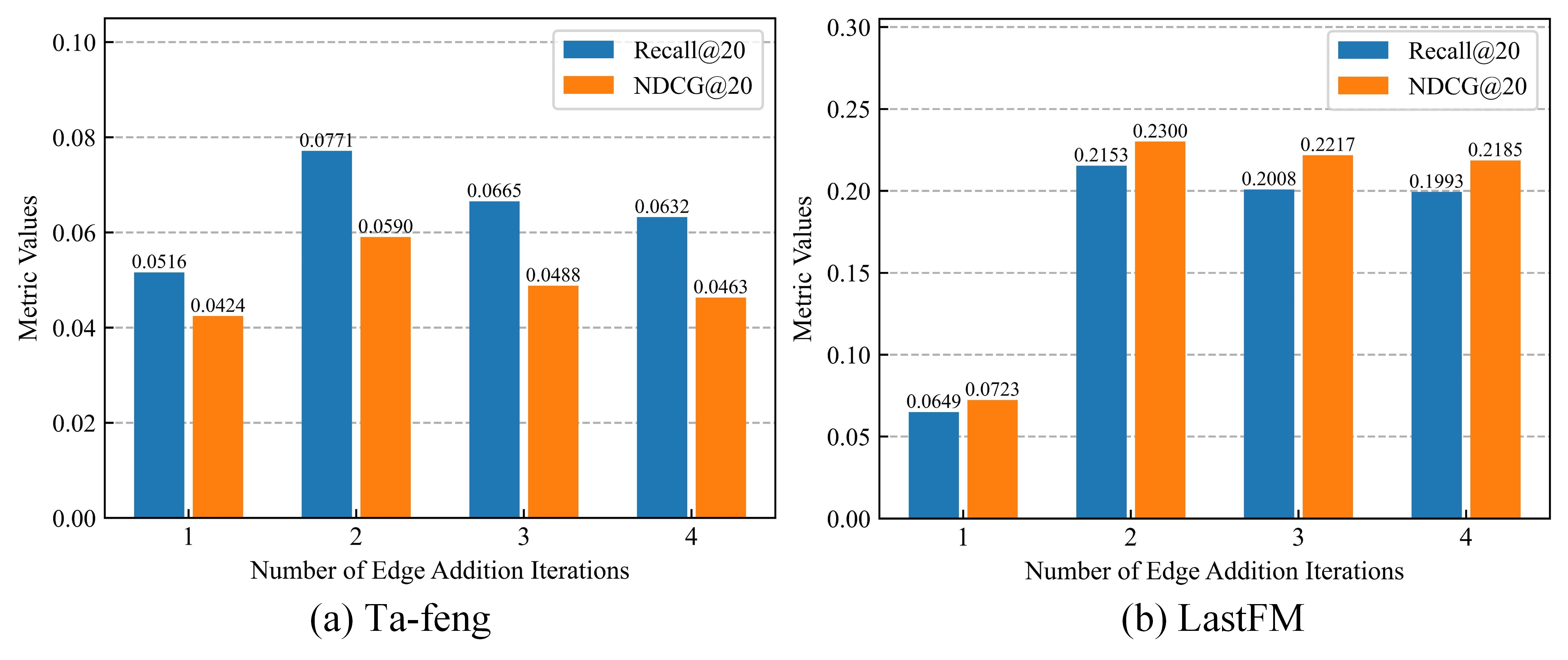}
	
	\vspace{-0.1cm}
	
	\caption{Impact of Positive Sample Distribution.}
	
	\vspace{-0.5cm}
	
\end{figure}

Considering the potential issue of positive sample distribution imbalance caused by TFPS's layer-enhancement strategy for constructing the positive sample set, we investigated the impact of the positive sample distribution on model performance by controlling the number of edge addition iterations. When the iterations is 1, the positive sample set remains the original set; when it is 2, it corresponds to TFPS's layer-enhancement strategy on Ta-feng and Lastfm datasets; and when it exceeds 2, the positive sample set is further augmented with additional positive samples reflecting users' current preferences. 
As shown in Fig.~5, moderate iterations are beneficial, whereas excessive repetition may over-concentrate the positive frequency distribution and reduce generalization, validating the effectiveness of TFPS’s linear layer-enhancement.

\begin{figure}[t]
	\centering
	\includegraphics[width=\linewidth]{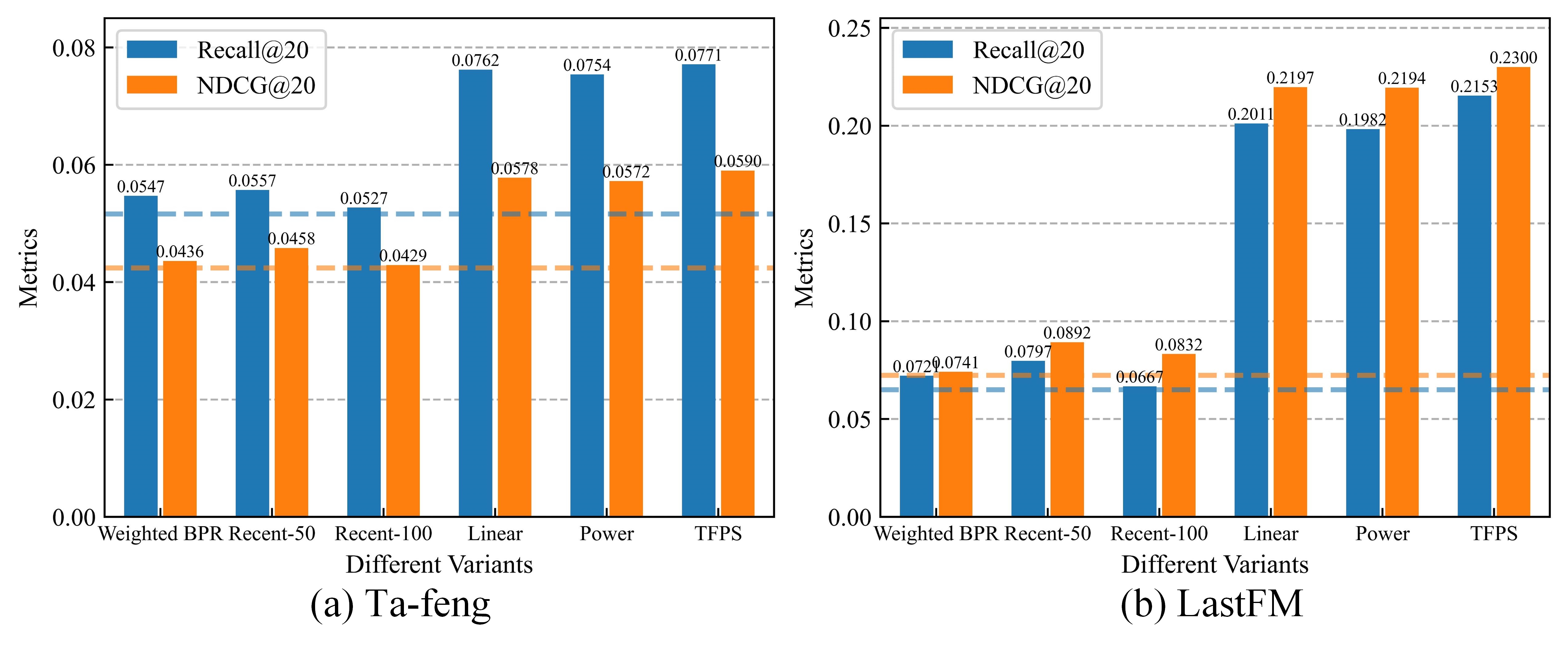}
	
	\vspace{-0.1cm}
	
	\caption{Impact of different temporal strategies.}
	
	\vspace{-0.2cm}
	
\end{figure}

\begin{figure}[t]
	\centering
	\includegraphics[width=\linewidth]{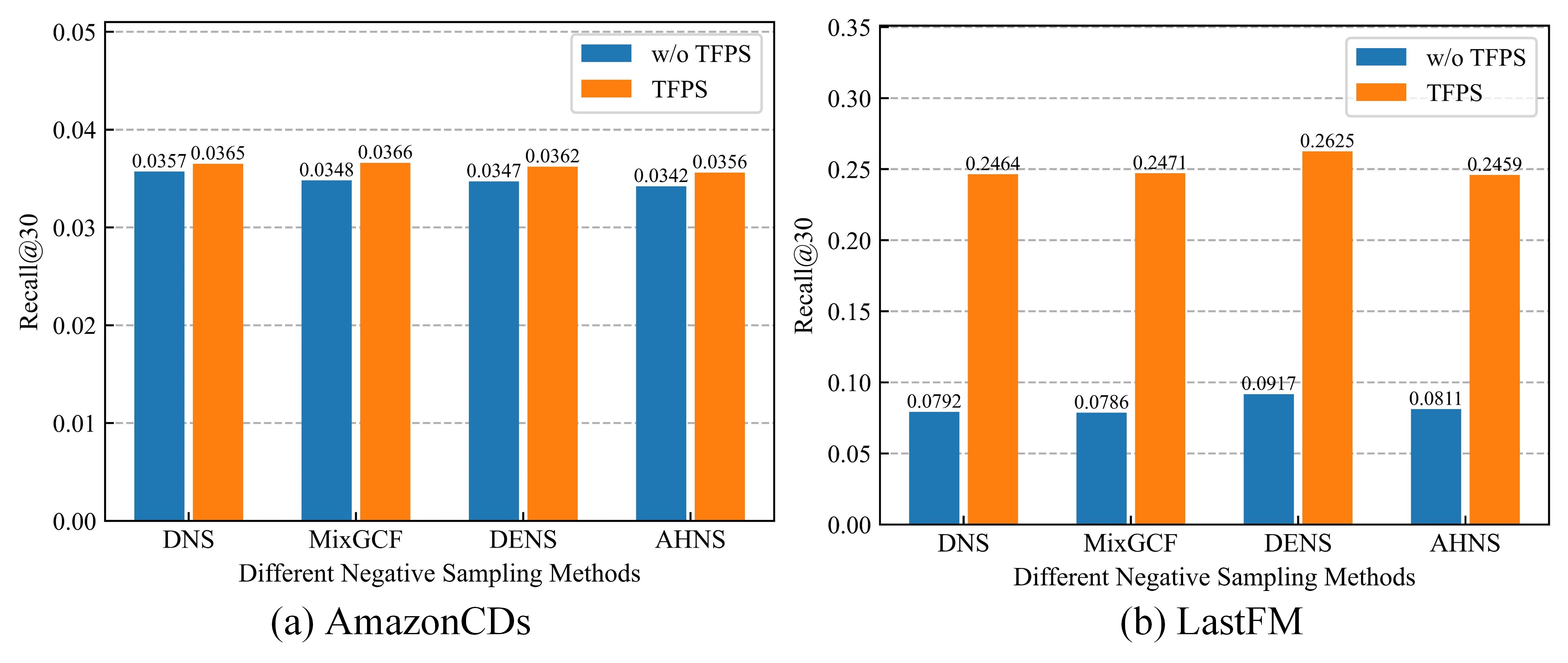}
	
	\vspace{-0.1cm}
	
	\caption{Performance comparison of integrating TFPS with different negative sampling methods.}
	
	\vspace{-0.4cm}
	
\end{figure}

To further verify the effectiveness of the exponential time-decay function and the layer-enhancement mechanism in TFPS, we compare TFPS with the following variants: 
(1) \textit{Weighted-BPR}, which applies time-decay weights directly to each positive instance in the BPR loss;
(2) \textit{Recent-$k$}, which constructs positives using only the $k$ most recent interactions of each user; and
(3) \textit{Linear} and \textit{Power}, which replace the exponential decay with linear and power-law decay functions, respectively.
For all variants, hyperparameters are tuned to ensure a fair comparison.

As shown in Fig.~6, the blue and orange dashed horizontal lines denote the backbone (LightGCN) performance in terms of Recall@20 and NDCG@20, respectively. Exponential decay consistently outperforms linear and heavy-tailed power-law decays, which aligns with the intuition that exponential decay can impose a stronger recency bias for capturing time-evolving preferences. 
Moreover, recency-aware strategies outperform LightGCN, indicating that emphasizing recent interactions helps the model learn users' current preferences and better predict future behaviors. TFPS further improves upon these recency-aware baselines. 
We attribute the gain to the linear layer-enhancement mechanism: it amplifies interactions that better reflect current preferences when constructing the positive sample set, while still retaining those representing older preferences. This design ensures that current preferences remain the dominant learning signal while leveraging informative cues from past interactions (e.g., long-term stable interests and latent preference patterns), which helps the model learn a more comprehensive representation of user interests and thereby make more accurate future preference predictions.

\subsection{Application Study (RQ3)}

In this section, we further examine the applicability of TFPS by integrating its constructed positive sample sets with various negative sampling methods and implicit CF models.

Fig.~7 shows the performance comparison of integrating TFPS into other negative sampling methods on the AmazonCDs and LastFM datasets. It can be seen that TFPS improves the effectiveness of negative sampling methods, proving its strong applicability. Additionally, we observe that integrating TFPS into these methods results in varying degrees of performance improvement for the model. We speculate that this is because most negative sampling methods typically obtain hard negative samples by combining positive samples (e.g., sampling negative samples based on similarity to positive samples, or constructing negative samples through mixup with positive samples). However, the positive sample set constructed by TFPS contains a higher proportion of recent user interaction, which reduces the diversity of negative samples, thus limiting further performance improvements of this model.

\begin{table}[!]
	\centering

	\setlength{\tabcolsep}{3.5pt} % 数值越大，列间距越宽
	
	\small
	
	\caption{Performance comparison integrated with MF.}

	\begin{tabular}{ccccc}
		\toprule
		Dataset      & \multicolumn{2}{c}{AmazonCDs} & \multicolumn{2}{c}{LastFM} \\ 
		\midrule
		Method        & Recall@30        & NDCG@30  & Recall@30        & NDCG@30  \\ 
		\midrule
		RNS         & \underline{0.0309}           & 0.0150         & 0.0868               & 0.0736        \\ 
		DNS           & 0.0305           & \underline{0.0154}         & 0.0783               & 0.0720        \\ 
		DNS(M,N)      & 0.0293           & 0.0148         & 0.0754               & 0.0693        \\ 
		DENS          & 0.0259           & 0.0118         & 0.0838               & 0.0748        \\ 
		AHNS          & 0.0253           & 0.0122         & 0.0793               & 0.0578        \\ \midrule
		R-CE       &   0.0274   &     0.0133       &     0.0904      &    0.0792  \\ 
		T-CE       &   0.0299         &   0.0146    &    0.0913        &   \underline{0.0813}   \\ 
		DeCA       &   0.0214         &   0.0101    &    0.0747        &   0.0634   \\ 
		DCF        &    0.0273        &   0.0133  &      \underline{0.0936}             &     0.0783       \\
		PLD &         0.0130&   0.0066&  0.0295&   0.0237  \\ \midrule
		TFPS          & \textbf{0.0336}           & \textbf{0.0167}         & \textbf{0.2242}               & \textbf{0.2080}        \\ 
		\bottomrule
	\end{tabular}
	
	\vspace{-0.2cm}
	
\end{table}

Table III shows the experimental results of TFPS and other baselines on the MF-based CF recommendation model \cite{30}. Among them, we omit MixGCF, STAM, TFPS-STAM, as these methods, designed based on graph neural network aggregation processes, are not applicable to MF-based CF models. It can be observed that the Recall@30 and NDCG@30 of TFPS are significantly higher than those of other baselines. This performance benefits from our data-centric approach, which constructs a high-quality positive sample set based on temporal information, explicitly guiding the CF recommendation model to capture users’ current preferences. This further demonstrates the strong applicability of TFPS. Additionally, PLD also exhibits particularly poor performance on this backbone, highlighting the limitations of its resampling strategy in capturing users’ current preferences, namely the limited coverage of sampled positive examples and the high proportion of outdated interactions.

\subsection{Comparison with Sequential Recommendation Models (RQ4)}

In this section, we provide a reference comparison with sequential recommendation models to further examine the ability of TFPS in capturing temporal information. Since sequential and implicit CF models differ in both task objectives and data modeling paradigms (see Section~II.B), this comparison is not a strict apples-to-apples evaluation. To align the evaluation, we treat the sequence embedding as an equivalent representation of the user embedding and assess whether the ground-truth item in the test set appears in the top-$K$ predictions.
We compare TFPS with six representative or state-of-the-art sequential recommendation models, including the RNN-based model (GRU4Rec \cite{60}), transformer-based models (BERT4Rec \cite{61}, LinRec \cite{62} and FEARec \cite{63}), and SSM-based models (ECHOMamba \cite{64} and SIGMA \cite{65}). The hyperparameter settings of these models are kept consistent with those used in the official implementation of \cite{65}\footnote{https://github.com/Applied-Machine-Learning-Lab/SIMGA}. As shown in Table IV, TFPS outperforms sequential recommendation models in most cases. This improvement can be attributed to the fact that TFPS explicitly models temporal interval information, whereas these sequential models only capture the interaction order. The above results demonstrate that TFPS effectively incorporates temporal information into negative sampling, thereby enhancing model performance.

\begin{table}[!]
	\centering
	
	\setlength{\tabcolsep}{3.5pt} % 数值越大，列间距越宽
	
	\small
	
	\caption{Performance comparison between TFPS and sequential recommendation models.}
	
	\begin{tabular}{ccccc}
		\toprule
		Dataset      & \multicolumn{2}{c}{Ta-feng} & \multicolumn{2}{c}{LastFM} \\ 
		\midrule
		Method        & Recall@30        & NDCG@30  & Recall@30        & NDCG@30  \\ 
		\midrule
		GRU4Rec       & 0.0736 & 0.0530 & 0.2020 & 0.1928 \\
		BERT4Rec    & 0.0795 & 0.0595 & 0.2044 & 0.1980 \\
		LinRec         & 0.0812 & 0.0606 & 0.2218 & 0.2169 \\
		FEARec         & 0.0755 & 0.0546 & 0.2173 & 0.2164 \\
		ECHOMamba & 0.0787 & 0.0568 & 0.1964 & 0.1920 \\
		SIGMA           & \underline{0.0834} & \underline{0.0612} & \textbf{0.2502} & \underline{0.2314} \\
		\midrule
		TFPS           & \textbf{0.0915} & \textbf{0.0638} & \underline{0.2485} & \textbf{0.2395} \\
		\bottomrule
	\end{tabular}
	
	\vspace{-0.2cm}
	
\end{table}

\section{CONCLUSION}

In this paper, we focus on the construction of the high-quality positive sample set in negative sampling and propose a temporal filtration-enhanced method, TFPS. The method incorporates temporal information into the bipartite graph through a time decay model and extracts the user's current preferences from the graph using graph filtration. A layer enhancement mechanism is then used to build a high-quality positive sample set, guiding the model to more accurately learn the user's current interests based on these supervisory signals. Theoretical analysis and extensive experiments demonstrate that TFPS can achieve higher improvements in implicit CF recommendation models. This also provides a new research direction for negative sampling strategies, which can further improve the performance of implicit CF models.

\bibliographystyle{ACM-Reference-Format}
\bibliography{citiation}

\end{document}